\documentclass[11pt]{article}

\usepackage{amsmath, amssymb, amsthm, mathtools}
\usepackage{amsmath,amssymb,amsthm,amsfonts}
\usepackage{dsfont}
\usepackage{bbm}
\usepackage[nice]{nicefrac}
\usepackage{graphicx}      
\usepackage{float}         
\usepackage{booktabs}      
\usepackage{multirow}      
\usepackage{array}         

\usepackage{subfig}

\usepackage{url}           
\usepackage{hyperref}      
\hypersetup{hyperfootnotes=false}
\hypersetup{
    hyperfootnotes=false,
    colorlinks=true,        
    citecolor=blue,         
    pdfborder={0 0 0}       
}
\usepackage[round]{natbib} 
\usepackage{bibentry}      

\usepackage{enumitem}      

\usepackage{tikz}          
\usetikzlibrary{quotes, positioning}

\usepackage[letterpaper]{geometry} 
\newtheorem{definition}{Definition}[section]
\newtheorem{theorem}{Theorem}
\newtheorem{lemma}{Lemma}[section]
\newtheorem{remark}{Remark}[section]

 \newtheorem{asm}{Assumption}

\newcommand{\1}{\mathbf{1}}

\title{Causal Feature Learning in the Social Sciences}

\author{
    Jingzhou Huang\thanks{Department of Economics and Quantitative Theory and Methods, Emory University.}  \and
    Jiuyao Lu\thanks{Department of Statistics and Data Science, The Wharton School, University of Pennsylvania.}  \and
    Alexander Williams Tolbert\thanks{Department of Quantitative Theory and Methods, Emory University.} 
}

\begin{document}

\maketitle

    \begin{abstract}
        Variable selection poses a significant challenge in causal modeling, particularly within the social sciences, where constructs often rely on inter-related factors such as age, socioeconomic status, gender, and race. Indeed, it has been argued that such attributes must be modeled as macro-level abstractions of lower-level manipulable features, in order to preserve the modularity assumption essential to causal inference (\cite{mossé2025modelingdiscriminationcausalabstraction}). This paper accordingly extends the theoretical framework of Causal Feature Learning (CFL). 
        Empirically, we apply the CFL algorithm to diverse social science datasets, evaluating how CFL-derived macrostates compare with traditional microstates in downstream modeling tasks. 
    \end{abstract}

\section{Introduction}
\label{sec:introduction}

When an individual is discriminated against, their worse treatment is partly explained by the fact that they possess a protected attribute, such as race or gender (\cite{KlemThomsen2018}). Where explanation is understood causally, we arrive at a causal fairness constraint, which restricts causal pathways between individuals' protected attributes and their treatment \cite{loi2023would}. Causal algorithmic fairness criteria provide different ways of making this constraint precise (\cite{kilbertus2017avoiding,kusner2017counterfactual,barocas-hardt-narayanan}). For example, \textit{counterfactual fairness} compares the distribution of a predictor $\hat{Y}$, given an individual's features $x$ and their protected attribute $a$ (e.g. a race or gender), to the counterfactual distribution of the predictor, where that individual's protected attribute takes a different attribute $ a^\prime$  (\cite{kusner2017counterfactual}):\footnote{For comparison between metric fairness criteria (\cite{dwork2012fairness}) and causal criteria, see (\cite[pp.90-97]{Plecko2024}). For comparison between so-called ``statistical'' criteria and causal criteria, see \cite{glymour2019measuring,beigang2023reconciling}.} 
\begin{align}
    P(\text{do}(a)\, \hat{Y}=y\,|\, x, a) = P(\text{do}(a^\prime)\, \hat{Y}=y\,|\, x, a).\label{eq:counterfactual fairness}
\end{align}
Of course, estimating the quantities that feature in causal criteria, e.g. (\ref{eq:counterfactual fairness}), requires substantive assumptions, but this is a general feature of causal inference, and theorists have developed methods for handling uncertainty about the underlying causal structure (\cite{russell2017worlds}). A more basic problem for causal notions of fairness is to clarify the sense in which an individual's race, gender, or other protected attribute can \textit{cause} their worse treatment.

An initial challenge contended that protected attributes cannot cause outcomes, because they are not manipulable (cf. \cite{holland1986statistics}, p. 946; \cite{greiner2011causal}, pp. 1-2; \cite{glymour2014race}; \cite{sen2016race}, p. 504). However, the impossibility of intervening on protected attributes is not necessarily a problem (\cite{pearl2018does}). After all, it is impossible to intervene on people's smoking habits, but we can estimate the effects of smoking (\cite[p. 1268]{weinberger22}). A more recent challenge contends that protected attributes like race cannot be modeled as causes of worse treatment, because the constitutive relations between these attributes and individuals' other features violate the \textit{modularity assumption} essential to causal inference, which requires the ability to isolate causal effects (\cite{kohlerhausmann2018eddie, hu2020sex, hu2022causation, hu2024}).

In reply, theorists have suggested that modularity can be preserved, if the constitutive relations between race and other attributes are modeled using the framework of causal abstraction (\cite{mossé2025modelingdiscriminationcausalabstraction}), i.e. protected attributes like race are modeled as macrostates, which correspond systematically to the states of micro-variables that constitute them. While this answers the modularity worry in theory, the approach is empirically untested. Heterogeneity among micro-variables is a general obstacle to the construction of macrostates \cite{spirtes2004causal}, and theorists have given principled reasons for thinking that this problem may arise for macrostates corresponding to protected attributes (\cite{tolbert2024causal,tolbert2024restricted}). In sum, it is an open empirical question whether one can construct protected attributes like age, race, and gender as macrostates, in a way that preserves the causal profile of the corresponding micro-variables.

Variable selection poses a significant challenge in causal modeling, particularly within the social sciences, which rely heavily on causal inference (\cite[p. 650]{imbens2024causal}), and which often study complex, inter-related factors such as age, socioeconomic status, gender, and race. At the same time, the practical impossibility of experimentation often requires that researchers working in the social sciences rely on merely observational data (\cite[pp. 25-29]{gangl2010causal}), which in general greatly under-determines the underlying causal structure, as formalized in the so-called ``causal hierarchy theorems'' (\cite{ibeling2021topological,bareinboim2022pearl}).\footnote{We follow much of the algorithmic fairness literature in using Structural Causal Models. See \cite{ibeling2024comparing} for a partial equivalence between this framework and potential outcomes. See \cite{pearl2010} for a defense of Structural Causal Models in sociology.} Some have proposed that social scientists conjecture causal structure using qualitative data, for example by relying on ethnographies that describe how existing social practices produce outcomes of interest (\cite[pp. 67-70]{steel2004social}).

Other fields, including healthcare, genetics, and climatology, have turned to \textit{Causal Feature Learning} (CFL) (\cite{chalupka2014visual,chalupka2016unsupervised,chalupka2016multi,chalupka2017causal}), which addresses the variable selection challenge by identifying features that sustain causal relationships under various manipulations. CFL operates by partitioning the data into clusters of macrostates that encapsulate the essential causal dynamics among microstates. This allows the data to determine how variables are constructed, uncovering latent structures that predefined categories may overlook. Formally, CFL seeks to identify a partition of the microstates into macrostates such that for each macrostate, the causal effect on the outcome is preserved. That is, given a set of microstate causes $\mathcal{X}$ and microstate effects $\mathcal{Y}$, CFL aims to find a partition $\Pi(\mathcal{X})$ such that for any two microstates $x_1, x_2 \in \mathcal{X}$, if $x_1 \sim x_2$ under $\Pi(\mathcal{X})$, then for any macrostate $y \in \mathcal{Y}$:
\begin{align}
    P(Y \mid \operatorname{do}(x_1)) = P(Y \mid \operatorname{do}(x_2)).\label{eq:causal-equivalence}
\end{align}
This equivalence ensures that the macrostates derived from $\Pi(\mathcal{X})$ maintain the causal relationships necessary for accurate causal inference. However, CFL assumes that all variables are discrete, and is not directly applicable to many social scientific data sets. It is an open question how this theoretical framework can be extended to the continuous setting.

This paper provides answers to these theoretical and empirical questions. On the theory side, we show that a simple binning technique immediately extends the Causal Coarsening Theorem (CCT) of \cite{chalupka2017causal} to the continuous setting:

\begin{theorem}[Extended CCT, informal] A partition of continuous variables based purely on observational data can be refined into the partition defined by (\ref{eq:causal-equivalence}); the subset of distributions for which this fails has Lebesgue measure zero.
\end{theorem}

Because equation~\ref{eq:causal-equivalence} is a strict requirement, the partition guaranteed by the CCT may be trivial when each micro-state has a distinctive effect. Further, a natural question is whether non-trivial partitions exit, for which equation~\ref{eq:causal-equivalence} holds only approximately. We show that this is indeed the case:

\begin{theorem}[Regularity, informal]
    There always exist partition on micro-states $\Pi(\mathcal{X})$ and $\Pi(\mathcal{Y})$ for which the equation~\ref{eq:causal-equivalence} holds approximately.
\end{theorem}

After making theoretical extensions to the framework of CFL, we apply the extended CFL algorithm to diverse social science datasets, evaluating how CFL-derived macrostates compare with traditional microstates in downstream modeling tasks.\footnote{See \href{https://github.com/Rorn001/causal-learning.git}{https://github.com/Rorn001/causal-learning.git} for replication code and data} We show that CFL reliably reduces dimensionality while detecting heterogeneity in the effects of macrostates.

\textbf{Roadmap:} Section~\ref{sec:preliminaries} introduces the definitions and notation that feature in the paper. Section~\ref{sec:CFL-discrete} introduces the framework of Causal Feature Learning and shows the above theorems. Section~\ref{sec:datasets} applies the CFL algorithm to several social scientific datasets. 

\section{Preliminaries and Notation}
\label{sec:preliminaries}

We follow standard notation and terminology from the causal inference literature. Let $\mathcal{X}$ represent the set of \textbf{microstates}, with $x \in \mathcal{X}$ denoting the individual covariates or features of interest in a dataset. These microstates represent the most granular elements in the data, such as age, education, income, or other social indicators in the context of social sciences. Let $\mathcal{Y}$ represent the set of \textbf{outcomes}, with $y \in \mathcal{Y}$ denoting the observed effects or responses (e.g., voting behavior, income after a treatment, etc.).

Additionally, we assume the presence of treatment variables, where the intervention $\operatorname{do}(\mathcal{X}=x)$ sets the microstate $\mathcal{X}$ to a specific value $x$ through a treatment or policy intervention. Our primary goal is to develop methodologies that construct \textbf{macrostates} from the microstates. A \textbf{macrostate} is a coarser, aggregated version of the microstates, grouped in such a way that the underlying causal structure is preserved. Formally, a macrostate is a partition $\Pi(\mathcal{X})$ of the set of microstates into clusters, where each cluster summarizes multiple microstates while maintaining their causal effect on the outcome variable $\mathcal{Y}$.

In this paper, we extend the Causal Feature Learning (CFL) framework to automatically construct these macrostates, particularly when dealing with both discrete and continuous microstates. The goal is to find partitions of $\mathcal{X}$ into macrostates that preserve the causal relationships necessary for accurate causal inference, specifically ensuring that for any two microstates $x_1$ and $x_2$ in the same partition, their causal effect on microstate effect $\mathcal{Y}$ remains equivalent under manipulation:
\[
P(Y \mid \operatorname{do}(x_1)) = P(Y \mid \operatorname{do}(x_2)).
\]
This ensures that the derived macrostates are both informative and causally consistent with the microstates they summarize.

\section{Causal Feature Learning}
\label{sec:CFL-discrete}

This section provides an overview of the CFL framework, and shows Theorem~1.

\subsection{Preliminaries and Assumptions}
\label{subsec:preliminaries}

\begin{definition}[Microstates]
\label{def:micro}
Let $\mathcal{X}$ denote the set of microstate causes and $\mathcal{Y}$ the set of microstate effects. The variable $z$ denotes a confounder influencing both $x \in \mathcal{X}$ and $y \in \mathcal{Y}$.
\end{definition}

\begin{asm}[Discrete Macrostates]
\label{assumption:discrete}
All macrostates are discretized into a finite number of states, denoted by $\mathcal{I}$. This ensures a finite state space, facilitating the application of clustering methods for data partitioning.
\end{asm}

\begin{asm}[Smoothness]
\label{assumption:smoothness}
The conditional distribution $P(\mathcal{Y} \mid \mathcal{X})$ may exhibit discontinuities at the boundaries of observational states. However, a continuous density function $f: \mathcal{X} \times \mathcal{Y} \to \mathbb{R}$ exists such that $P(y \mid x) = f(x, y)$, where $f$ is piecewise continuous.
\end{asm}

\begin{definition}[Partitions]
\label{def:partitions}
The observational partition $\Pi_o(\mathcal{X})$ with respect to $\mathcal{Y}$ is induced by the equivalence relation:
\[
x_1 \sim x_2 \iff \text{for all } y \in \mathcal{Y}, \ P\left(Y \mid x_1\right) = P\left(Y \mid x_2\right).
\]
The fundamental causal partition $\Pi_c(\mathcal{X})$ with respect to $\mathcal{Y}$ is induced by the equivalence relation:
\[
x_1 \sim x_2 \iff \text{for all } y \in \mathcal{Y}, \ P\left(Y \mid \operatorname{do}\left(x_1\right)\right) = P\left(Y \mid \operatorname{do}\left(x_2\right)\right).
\]
The confounding partition $\Pi_{P(X \mid Z)}(\mathcal{X})$ with respect to $\mathcal{Z}$ is induced by the equivalence relation:
\[
x_1 \sim x_2 \iff \text{for all } z\in\mathcal{Z},  \ P\left(x_1 \mid z\right) = P\left(x_2 \mid z\right) .
\]
\end{definition}



\begin{definition}[Macrostate Manipulation]
\label{def:macro-manipulation}
The operation $do(I=i)$ on a macrostate $I$ is defined by manipulating the underlying microstate to a value $x_k$ such that $X(x_k)=i$. Formally, $do(I=i)$ is expressed as $do(X(x_k)=i)$ for $I=i$. 
\end{definition}


\subsection{Causal Coarsening Theorem (CCT)}
\label{sec:CCT-discrete}

\cite{chalupka2017causal} prove that the causal partition will (almost always) be a coarsening of the observational partition, which justifies applying clustering algorithm to observed data to uncover the underlying causal structure.

\begin{theorem}[Causal Coarsening Theorem]
\label{thm:cct}
Consider the set of joint distributions $P(X, Y, Z)$ that induce a fixed causal partition $\Pi_c(\mathcal{X})$ and a fixed confounding partition $\Pi_{P(X \mid Z)}(\mathcal{X})$. The subset of distributions where $\Pi_c(\mathcal{X})$ is not a coarsening of $\Pi_o(\mathcal{X})$ has Lebesgue measure zero.
\end{theorem}



\begin{remark}
    Analogous to partitioning $\mathcal{X}$, \cite{chalupka2017causal} define the observational partition $\Pi_o(\mathcal{Y})$ and causal partition $\Pi_c(\mathcal{Y})$ of $\mathcal{Y}$. They used the same teniniques stated above to prove the following statement: The subset of distributions where $\Pi_c(\mathcal{Y})$ is not a coarsening of $\Pi_o(\mathcal{Y})$ has Lebesgue measure zero.
\end{remark}

Instead of considering point mass of the conditional distribution $Y \mid X$, we consider the probability of $\mathcal{Y}$ falling into certain interval while conditioning on $\mathcal{X}$. More formally, we discretize $\mathcal{Y}$ into $m$ bins of the form $(a_k,a_{k+1}]$, where $k=1,2,\dots,m$. Such discretization induces the observational partition and causal partition for continuous variables.

\begin{definition}[Partitions with Binning]
\label{def:partition-binning}
For continuous macrostates, the observational partition $\Pi_o(\mathcal{X})$ is induced by the equivalence relation:
\begin{align*}
x_1 \sim x_2 \iff  &\text{ for all bins } (a_k, a_{k+1}],\\  &P\left(a_k < \mathcal{Y} \le a_{k+1} \mid x_1\right) = P\left(a_k < \mathcal{Y} \le a_{k+1} \mid x_2\right).
\end{align*}
Analogously, the causal partition $\Pi_c(\mathcal{X})$ with binning is induced by the equivalence relation:
\begin{align*}
x_1 \sim x_2 \iff  &\text{ for all bins } (a_k, a_{k+1}],\\& \ P\left(a_k < \mathcal{Y} \le a_{k+1} \mid \operatorname{do}(x_1)\right) = P\left(a_k < \mathcal{Y} \le a_{k+1} \mid \operatorname{do}(x_2)\right).
\end{align*}
\end{definition}

The binning process involves partitioning continuous variables into discrete intervals, allowing CFL to apply in scenarios where variables are not naturally categorical. The choice of the number of bins $m$ and the binning thresholds $(a_k, a_{k+1}]$ can be guided by statistical methods such as equal-width binning, equal-frequency binning, or more sophisticated techniques like entropy-based binning to optimize the preservation of causal relationships. With the binning technique, the Causal Coarsening Theorem extends to continuous variables straightaway:

\begin{theorem}[Extended Causal Coarsening Theorem]
\label{thm:cct-extended}
Consider the set of joint distributions $P(X,Y,Z)$ that induce a fixed causal partition $\Pi_c(\mathcal{X})$ and a fixed confounding partition $\Pi_{P(X \mid Z)}(\mathcal{X})$. Under the binning technique defined in Definition \ref{def:partition-binning}, the subset of distributions where $\Pi_c(\mathcal{X})$ is not a coarsening of $\Pi_o(\mathcal{X})$ has Lebesgue measure zero.
\end{theorem}

The full proof is relegated to Appendix \ref{app:proof-sec-CFL-continuous}. Here we include a sketch: 

\begin{proof}[Proof Sketch]
\label{proof:sketch}
Define the following probabilities:
\begin{align*}
\alpha_{k, x, z} & \triangleq P(\alpha_k < \mathcal{Y} \le \alpha_{k+1} \mid x, z), \\
\beta_{x, z} & \triangleq P(x \mid z), \\
\gamma_z & \triangleq P(z).
\end{align*}
Fixing $\beta_{x, z}$ directly fixes $\Pi_{P(X \mid Z)}(\mathcal{X})$, and fixing $\alpha_{k, x, z}$ directly fixes $\Pi_{P(Y \mid X, Z)}(\mathcal{Y})$, effectively fixing $\Pi_c(\mathcal{X})$.

The conditions $O(x_1) = O(x_2)$ can be expressed as equations based on $(\alpha,\beta,\gamma)$, which turn out to impose polynomial constraints on the joint distribution $P[\gamma ; \alpha, \beta]$. These constraints are non-trivial, as demonstrated by constructing specific examples where they do not hold. Therefore, the subset of $P[\gamma ; \alpha, \beta]$ violating the theorem's conditions is of Lebesgue measure zero. Consequently, the subset of joint distributions $P(X,Y,Z)$ that violate the theorem is also of measure zero.
\end{proof}

\begin{remark}
    Analogously, we could define the observational partition $\Pi_o(\mathcal{Y})$ and causal partition $\Pi_c(\mathcal{Y})$ of $\mathcal{Y}$ with the binning technique. The same argument applies to the $\Pi_c(\mathcal{Y})$ and $\Pi_o(\mathcal{Y})$, ensuring that the subset of distributions where $\Pi_c(\mathcal{Y})$ is not a coarsening of $\Pi_o(\mathcal{Y})$ also has Lebesgue measure zero.
\end{remark}

\subsection{Approximate Partitions}

The causal coarsening theorem says that the causal partition refines the observational one. But it provides no guarantee regarding the size of the partition. At the extreme, if each state $x$ has a distinctive effect $y$, the partition will be trivial: no two states $x, y^\prime$ will belong to the same equivalence class. A natural question, then, is whether one can find a partition $\Pi(\mathcal{X})$ which ensures that states $x \in \mathcal{X}$ have effects which are individually close to the average within their cell:
\[
\bigg|P(y |x) - \sum_{x^\prime \sim X} P(y\,|\,x^\prime)\bigg| < \epsilon.
\]
 The objective is to maximize the minimum value of $\epsilon$ for which the above holds for all pairs $(x,y)$ (\cite{beckersApproximateCausalAbstraction2019a}). This may be relaxed to hold on average over all pairs $(x,y)$ (\cite{rischelCompositionalAbstractionError2021}),\footnote{For a very general statement of the objective, see \cite[p. 18]{geiger2024causalabstractiontheoreticalfoundation}.}  but here we consider a different relaxation: we can require that the average effects across any $\epsilon$-fraction of the variables in a cell be $\epsilon$-close to the average effects across the cell as a whole. In other words, where $W_x= \{x^\prime: x \sim x^\prime\}$, and similarly for $W_y$, we require that for disjoint $A,B \subseteq W_x$ with $\epsilon < \nicefrac{|A|}{|W_x|}, \nicefrac{|B|}{|W_y|}$,
\[
    \bigg|\sum_{x^\prime \in A, y^\prime \in B} P(x^\prime|y^\prime) - \sum_{x^{\prime\prime} \in W_x, y^{\prime\prime} \in W_y} P(x^{\prime\prime}|y^{\prime\prime}) \bigg| < \epsilon
    \]

As we now show, the existence of such partitions over $\mathcal{X}$ (and $\mathcal{Y}$), and their relationship to the ``randomness'' of the effects of $\mathcal{X}$ on $\mathcal{Y}$, follow from a result of \cite{csaba2014weighted}, who show (roughly) that a weighted graph can be partitioned into similarly sized cells, which have relatively stable average weights into each other. We first introduce the relevant notion of a partition:

\begin{definition}[Weighted graphs]
    A \textit{weighted graph} $G$ is a set of vertices $V(G)$ and a set of undirected edges $E(G)$, along with a weight function assigning a weight $w_e \in \mathbb{R}^{\geq0}$ to each edge $e\in E$.
\end{definition}

\begin{definition}[$\epsilon$-regularity]
    Let $A, B$ be disjoint subsets of $V(G)$. The pair $(A,B)$ is $\epsilon$-regular if for every $A^\prime \subseteq B, B^\prime \subseteq B$ with $\epsilon < \nicefrac{|A^\prime|}{|A|}, \nicefrac{|B^\prime|}{|B|}$,
    \[
    \bigg| \frac{w(A^\prime, B^\prime)}{|A^\prime||B^\prime|} - \frac{w(A, B)}{|A||B|}\bigg| < \epsilon.
    \]
\end{definition}

\begin{definition}[$\epsilon$-regular partitions]
    A weighted graph $G$ has an $\epsilon$-regular partition if its vertex set $V$ can be participated clusters $W_0,\dots,W_{\ell}$, such that
    \begin{itemize}
        \item 
        The clusters $W_i$ are equal in size (plus or minus one), and $\epsilon > \nicefrac{|W_i|}{|V|}$.
        \item 
        All but at most $\epsilon \ell^2$ of the pairs $(W_i,W_j)$ are $\epsilon$-regular.
    \end{itemize}
\end{definition}

\noindent The upper bound on the number of clusters needed to provide a partition of this kind depends on how ``random'' the graph is, in the following sense:

\begin{definition}[Quasi-randomness]
    A weighted graph $G$ with vertex set $V$ is $(D,\beta)$-quasi-random if for any disjoint $A, B \subseteq V $ such that $\beta < \nicefrac{|A|}{|V|},\nicefrac{|B|}{|V|}$,
    \[
    \frac{1}{D}<\frac{w(A,B)}{|A||B|} \bigg/ \frac{w(V,V)}{|V||V|}  < D,
    \]
    where $w(A,B)$ denote the sum of the weights of all edges with one endpoint in $A$ and the other in $B$.
\end{definition}

In other words, $G$ is $(D,\beta)$ quasi-random when subsets $A,B$ of sizes at least $\beta$ (relative to $V$) have an average edge density which is equal to that of the overall graph, up to a $D$ multiplicative factor. \cite{csaba2014weighted} show (p. 5):

\begin{lemma}[Weighted Regularity Lemma]\label{lemma:regularity}
    Let $D > 1$ and $\beta,\epsilon \in (0,1)$ such that $0 < \beta \ll \epsilon \ll 1 / D$ and let $L \geq 1$. If $G$ is a weighted $(D,\beta)$-quasi-random graph on $n$ vertices with $n$ sufficiently large depending on $\epsilon, L$, then $G$ admits a weighted $\epsilon$-regular partition into sets $W_0,\dots,W_\ell$ such that $L \leq \ell\leq C_{e,L}$ for some constant $C_{e,L}$.
\end{lemma}

The desired result follows immediately:

\begin{theorem}\label{theorem:regular partitions}
    For micro-causes $\mathcal{X}$ and effects $\mathcal{Y}$, let $W_x, W_Y$ for $x \in \mathcal{X}, Y \in \mathcal{Y}$ denote the equivalence classes of $x$ and $\mathcal{Y}$ under partitions $\Pi(\mathcal{X})$ and $\Pi(\mathcal{Y})$. There exist partitions $\Pi(\mathcal{X})$ and $\Pi(\mathcal{Y})$ which are $\epsilon$-regular, in the sense that for all $A, \subseteq W_x, B \subseteq W_y$ with $\epsilon < \nicefrac{|A|}{|W_x|}, \nicefrac{|B|}{|W_y|}$,
    \[
    \bigg|\sum_{x^\prime \in A, y^\prime \in B} P(x^\prime|y^\prime) - \sum_{x^{\prime\prime} \in W_x, y^{\prime\prime} \in W_y} P(x^{\prime\prime}|y^{\prime\prime}) \bigg| < \epsilon
    \]
\end{theorem}

\noindent Indeed, let $G$ be a bipartite graph, with vertices corresponding to microstates~$\mathcal{X}, \mathcal{Y}$. Define $w_{(x,y)} = P(y|x)$. Then apply Lemma~\ref{lemma:regularity} to obtain the desired partition. Note that Theorem~\ref{theorem:regular partitions} is in fact equivalent to Lemma~\ref{lemma:regularity}, restricted to bipartite graphs. For any such graph, we may scale down edge weights by a constant $C$ (the maximum edge weight), view the resulting weights as probabilities $P(y|x)$, apply Theorem~\ref{theorem:regular partitions}, and then obtain an $\epsilon^\prime$-regular partition, for $\epsilon^\prime = \epsilon / C$.

Lemma~\ref{lemma:regularity} generalizes Szemerédi's Regularity Lemma (\cite{szemeredi1975}), and a significant literature explores the algorithmic side of this lemma, e.g. \cite{FOX_LOVÁSZ_ZHAO_2017}. Nonetheless, it what follows, we deploy a straightforward generalization of the CFL algorithm, which we find reliably reduces dimensionality while detecting heterogeneous effects of macro-states.

\section{Social Science Applications}
\label{sec:datasets}

This section presents empirical results from applying the CFL algorithm to two prominent social science datasets: the National Supported Work (NSW) dataset and the Voting dataset. We demonstrate how CFL-derived macrostates enhance causal inference, uncover the heterogeneity in causal effects, and reduce the dimensionality of the causal analysis. In Appendix \ref{sec:redlining-data}, using the redlining dataset with census data, We also show that CFL can be incorporated with the canonical causal inference technique, such as propensity score matching, to combat issues due to the nature of the observational dataset. Appendix \ref{sec:binning-implementation} presents the results from binning implementation on the NSW dataset.

\textbf{Algorithmic Preliminaries:} The CFL algorithm constitutes a conditional density estimation by a neural network and a partition of the dataset by clustering observations based on the estimated conditional probability. Common clustering techniques are KMeans and DBSCAN (Density-Based Spatial Clustering of Applications with Noise), the first of which allows controlling for the number of clusters as a hyperparameter, while the second of which tunes the number of clusters as a model parameter (\cite{cfl2022}). We use KMeans for our clustering tasks in this paper allowing for more flexible visualization and a better understanding of the algorithm, though the aim of the CFL algorithm is to identify the optimal number of macrostate without supervision. 

\subsection{The NSW Dataset}
\label{sec:NSW}

The National Supported Work (NSW) dataset is a cornerstone in evaluating the effectiveness of job training programs. Originally utilized to assess the impact of a randomized labor training program on participants' future earnings in \cite{lalonde1986}, this dataset provides a robust framework for testing the validity of the CFL algorithm in social science data. By deriving macrostates from socioeconomic indicators, we aim to construct variables that better capture the causal pathways influencing earnings outcomes.


\begin{figure}[t]
    \centering
    \caption{Clustering of NSW Participants Based on Education, Age, and Treatment Assignment}
    \subfloat[2 clusters]{%
        \includegraphics[width=0.32\textwidth]{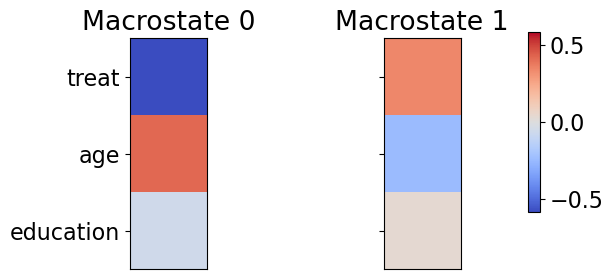}
        \label{fig:NSW_clusters_2}
    }
    \quad\quad
    \subfloat[3 clusters]{%
        \includegraphics[width=0.45\textwidth]{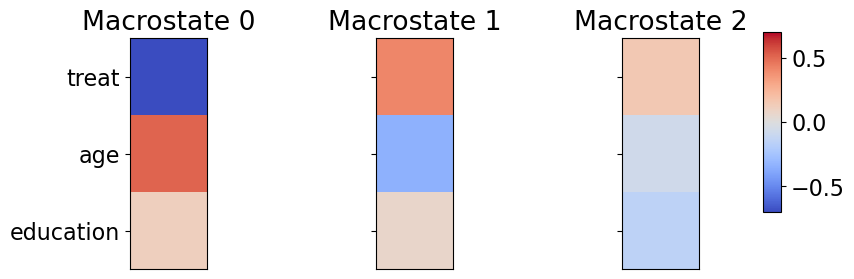}
        \label{fig:NSW_clusters_3}
        }
        \hfill
    \subfloat[4 clusters]{%
        \includegraphics[width=0.5\textwidth]{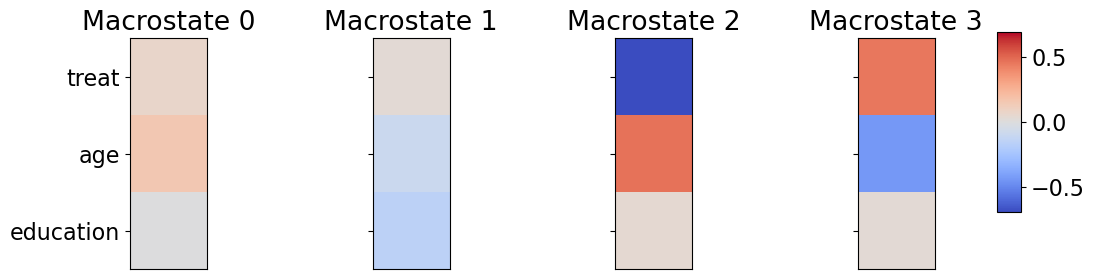}
        \label{fig:NSW_clusters_4}
    }
    \label{fig:NSW_clusters}
\end{figure}

\textbf{Clustering Analysis:} Figure \ref{fig:NSW_clusters} compares global (sample) averages and local (within-macrostate) averages of covariates across clusters constructed by the CFL algorithm. The variable \texttt{`treat'} is a treatment indicator. Varaible \texttt{`age'} and \texttt{`education'} are integers representing participants’ age and years of education, respectively. The outcome variable measures the change in income between pre- and post-treatment periods.

With two clusters, CFL partitions the sample into one macrostate predominantly composed of younger, treated individuals and another macrostate composed of older, untreated individuals. The local averages of education remain close to the global average, whereas the local averages of treatment and age differ significantly from global ones, suggesting that they drive the clustering process. 

As the number of clusters increases, CFL refines the partitioning, introducing greater granularity. While education initially plays a limited role in defining macrostate, it becomes more influential in the three- and four-cluster cases. A cluster primarily composed of untreated individuals remains stable across different specifications, whereas other groups are further divided into subgroups with distinct characteristics. 



\subsection{Causal Feature Learning Detects Heterogeneous Causal Effects}
\label{sec:heterogeneity}

Heterogeneity in causal effects is a prevalent phenomenon in social sciences, where the impact of an intervention varies across different subpopulations. The CFL algorithm facilitates the identification of such heterogeneity by uncovering macrostates that segment the population based on underlying causal mechanisms.
\begin{figure}[t]
    \centering
    \caption{Distribution of Treated and Untreated Units Across Clusters with Kernel Density Estimates}
    \includegraphics[scale=.2]{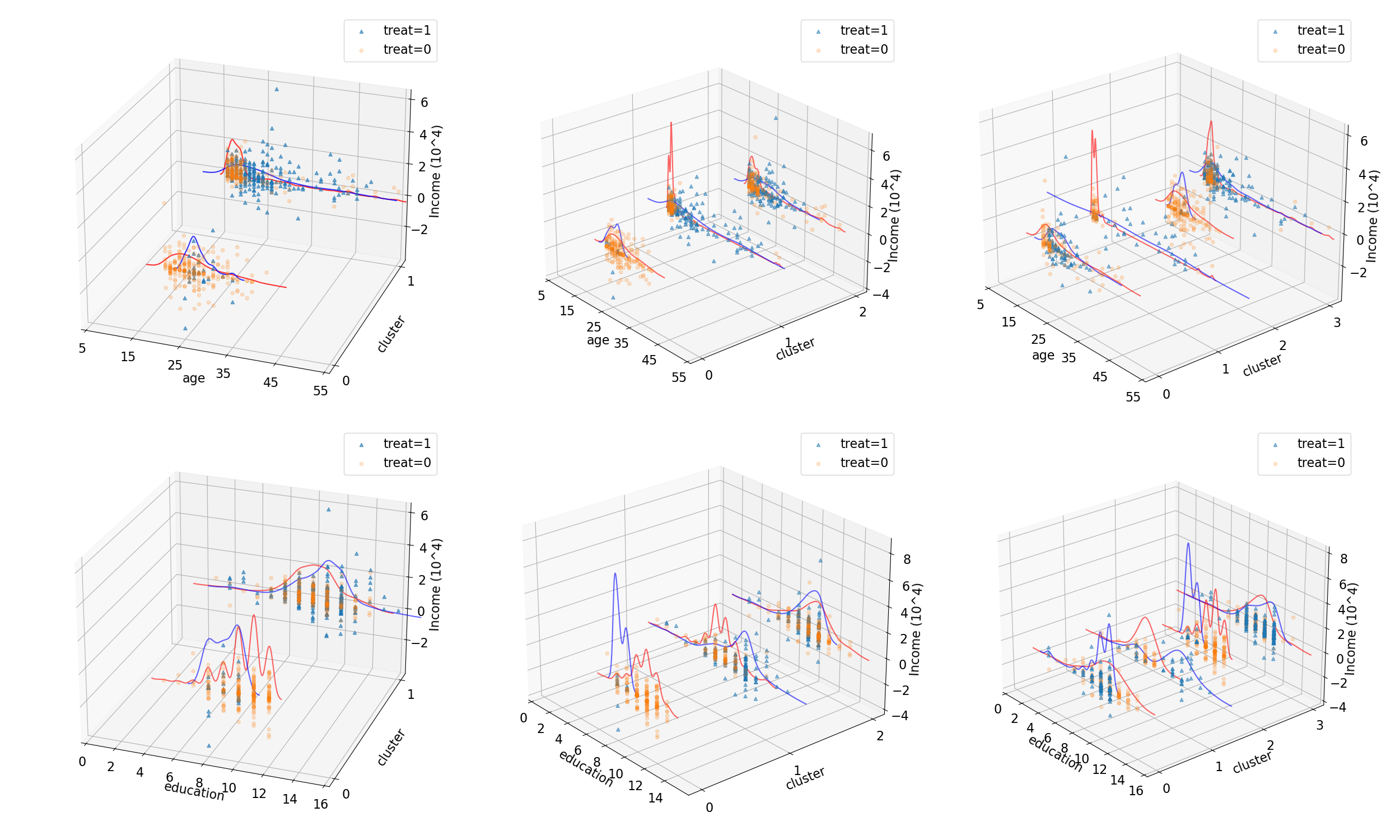}
    \label{fig:heterogeneity_clusters}
\end{figure}

\textbf{Cluster Distribution Across Ages:} Figure \ref{fig:heterogeneity_clusters} illustrates the distribution of treated and untreated units across age and education-based clusters. The CFL algorithm partitions the NSW participants into distinct subgroups based on these covariates and treatment status. Notably, younger participants (under 20 years old) predominantly cluster within the treated group, suggesting that CFL identifies them as having a comparable likelihood of achieving similar earnings outcomes, regardless of treatment status. A bifurcation around age 20 suggests a threshold beyond which treatment effects become more pronounced. For younger participants, the similarity in outcome probabilities between treated and untreated groups implies that external factors, such as baseline earning ability, may overshadow the treatment effect. This finding highlights the presence of heterogeneous treatment effects, with age acting as a moderating factor. Conversely, older participants exhibit a more balanced distribution between treated and untreated groups, suggesting that the treatment effect is sufficiently strong for CFL to distinguish between them. 

We can test the validity of the heterogeneity detected by CFL, from a more canonical perspective, by considering the following regression: 
\begin{align*}
Y_i = \beta_0 + \beta_1 \text{Treat}_i + \beta_2 \mathds{1}[age_i > \overline{age}] + \beta_3 \text{Treat}_i \mathds{1}[age_i > \overline{age}] + \epsilon_i
\end{align*}
The interaction term coefficient, $\beta_3$, captures the heterogeneity in treatment effects on earnings changes between individuals below and above the sample’s average age. Table \ref{tab:regression_results} shows that $\beta_3$ is large and statistically significant, indicating that younger individuals exhibit a weaker response to treatment, while older individuals experience significantly greater earnings gains post-treatment. These findings align with CFL clustering results shown above.

\textbf{Implications of Asymmetric Clustering:} The persistence of this asymmetric distribution, even as the number of clusters increases in Figure \ref{fig:heterogeneity_clusters}, reinforces the notion that the detected heterogeneity is not an artifact of inadequate clustering granularity. Instead, it reflects genuine variations in treatment effects across different age groups. The younger cohort may inherently possess characteristics that mitigate the treatment's impact leading to similar outcome distributions irrespective of treatment status. 

\textbf{Heterogeneity vs. Bias Considerations:} The asymmetric clustering in the NSW example raises important considerations regarding heterogeneity and potential biases. If the treatment assignment is not well randomized, it becomes challenging to discern whether the observed differences in treatment effects across subgroups indicate true heterogeneity or merely reflect biased sampling. For example, if younger people with inherently higher baseline earning ability, which could be unobservable, are more likely to remain untreated, the similar outcome distributions in Figure \ref{fig:heterogeneity_clusters} between treated and untreated young individuals could stem from this confounding relationship rather than the treatment's heterogeneous effectiveness.

\subsubsection{Rationale behind heterogeneity detection by CFL}

\begin{lemma}[Heterogeneity on Average]
    Assume randomized treatment, and let Y be an outcome of interest, $D$ be a treatment dummy variable, and $\mathcal{X}$ be a random vector for covariates. Further assume that the expectations of two variables are unequal as long as the distributions are unequal. If there exists some $i,j$ such that $(D=1, X=x_j) \sim (D=0, x=x_j)$ and $(D=1, X=x_i) \not\sim (D=0, X=x_i)$, or if there exists some $i, j$ such that $(D=1, X=x_i) \sim (D=1, X=x_j)$ and $(D=0, X=x_i) \not\sim (D=0, X=x_j)$, this implies the heterogeneity of treatment effect. 
\end{lemma}

\begin{proof}
Heterogeneity arises when
\begin{align*}
\mathbb{E}[Y \mid D=1, X=x_j] - \mathbb{E}[Y \mid D=0, X=x_j] \neq \mathbb{E}[Y \mid D=1, X=x_i] - \mathbb{E}[Y \mid D=0, X=x_i]
\end{align*}
for some $i, j$. If
\begin{align*}
P(Y \mid D=1, X=x_j) - P(Y \mid D=0, X=x_j) \neq P(Y \mid D=1, X=x_i) - P(Y \mid D=0, X=x_i),
\end{align*}
this can imply the previous inequality, but the converse may not necessarily be true. Thus, if CFL can detect some $i, j$ such that the above inequality holds,
CFL will manifest heterogeneity. There are two cases where this will hold:

1. For some $j$, $(D=1, X=x_j) \sim (D=0, x=x_j)$, which means they are clustered into one macrostate, so that
   \[
   P(Y \mid D=1, X=x_j) = P(Y \mid D=0, X=x_j),
   \]
   but for some $i$, $(D=1, X=x_i)$ is not in the same equivalence class as $(D=0, X=x_i)$, so
   \[
   P(Y \mid D=1, X=x_i) \neq P(Y \mid D=0, X=x_i),
   \]
   meaning that treatment has no effect on some subpopulation but has an effect on others.

2. For some $i, j$, $(D=1, X=x_i) \sim (D=1, X=x_j)$, so
   \[
   P(Y \mid D=1, X=x_j) = P(Y \mid D=1, X=x_i),
   \]
   but $(D=0, X=x_i)$ is not in the same equivalence class as $(D=0, X=x_j)$, so
   \[
   P(Y \mid D=0, X=x_i) \neq P(Y \mid D=0, X=x_j),
   \]
   meaning that $P(Y \mid D=1, X) - P(Y \mid D=0, X)$ is not constant across all values of $\mathcal{X}$.
    
\end{proof}

The assumption that unequal distribution implies unequal expectation may not always be true in the above lemma, so in this case, the CFL algorithm may not be valid to manifest heterogeneity on average. However, though heterogeneity on average is a common way to define the heterogeneous treatment effect, it is a summary statistic of the distribution, and the treatment can still be heterogeneous on the distribution level, if not on the expected level. Therefore, we can propose that it is sufficient for a treatment to be heterogeneous if it is heterogeneous on the distribution level according to the next lemma. 

\begin{lemma}[Heterogeneity on Distribution]
Assume randomized treatment, and let Y be an outcome of interest, $D$ be a treatment dummy variable, and $\mathcal{X}$ be a random vector for covariates. Heterogeneous treatment effect exists if $P(Y \mid D=1, X=x_j) - P(Y \mid D=0, X=x_j) \neq P(Y \mid D=1, X=x_i) - P(Y \mid D=0, X=x_i)$.
\end{lemma}

These arguments are similar and can account for the heterogeneity detected by CFL in the NSW dataset, where a portion of untreated young individuals are clustered into the same macrostate as the treated young individuals.

\begin{remark}
Randomized treatment is a necessary condition for the validity of CFL in detecting heterogeneity.
\end{remark}

Note that if the treatment is not randomly assigned, under the potential outcome framework, we can write a formal decomposition to show the existence of selection bias that will contaminate the detection of heterogeneity by CFL:
\begin{align*}
\mathbb{E}[Y(1) \mid D=1, X=x_j] - \mathbb{E}[Y(0) &\mid D=0, X=x_j] \neq \\
&\mathbb{E}[Y(1) \mid D=1, X=x_i] - \mathbb{E}[Y(0) \mid D=0, X=x_i]
\end{align*}
Plus and minus the unobserved term:
\begin{align*}
    \Rightarrow &\mathbb{E}[Y(1) \mid D=1, X=x_j] - \mathbb{E}[Y(0) \mid D=1, X=x_j] \\
    & \phantom{\mathbb{E}[Y(1) \mid D=1, X=x_j]} + \mathbb{E}[Y(0) \mid D=1, X=x_j] - \mathbb{E}[Y(0) \mid D=0, X=x_j] \\
    &\neq \mathbb{E}[Y(1) \mid D=1, X=x_i] - \mathbb{E}[Y(0) \mid D=1, X=x_i] \\
    &\phantom{\mathbb{E}[Y(1) \mid D=1, X=x_i]}+ \mathbb{E}[Y(0) \mid D=1, X=x_i] - \mathbb{E}[Y(0) \mid D=0, X=x_i]
\end{align*}
Rearrange:
\begin{align*}
    \Rightarrow &\mathbb{E}[Y(1) - Y(0) \mid D=1, X=x_j] \underbrace{ + \mathbb{E}[Y(0) \mid D=1, X=x_j] - \mathbb{E}[Y(0) \mid D=0, X=x_j]}_{\text{Selection Bias}} \\
    &\neq \mathbb{E}[Y(1) - Y(0) \mid D=1, X=x_i] + \underbrace{\mathbb{E}[Y(0) \mid D=1, X=x_i] - \mathbb{E}[Y(0) \mid D=0, X=x_i]}_{\text{Selection Bias}}
\end{align*}

Therefore, there could be no heterogeneity e.g. if $P(Y(1) - Y(0) \mid D=1, x=x_j) = P(Y(1) - Y(0) \mid D=1, x=x_i)$, but different degrees of selection biases across values of covariate will still result in the above inequality that manifests heterogeneity empirically through CFL. In other words, randomized treatment is a necessary condition for the validity of CFL in detecting heterogeneity.

\subsection{The Voting Dataset}
\label{sec:voting}

The Voting dataset, derived from the study by \cite{gerber2009does}, examines the ``persuasion effect" of offering a free subscription to the Washington Post (more liberal) or Washington Times (more conservative) on post-treatment political preference. In this randomized controlled trial, $n=3347$ individuals were assigned to receive either Washington Times ($Times_i=1$), Washington Post ($Post_i=1$), or control ($treat_i=0$). The outcome variable $Y_i$ we focused on is the broad policy index constructed by the authors to estimate the political attitude after the treatment, with higher index values indicating a more conservative stance. This dataset also includes a rich set of covariates encompassing demographic information, baseline political preferences, and historical voter turnout. \cite{jun2023identifying} utilized this dataset to explore persuasion effects, and we extend their analysis by applying the CFL algorithm to derive macrostates that may better capture the causal influence of the treatment on voting behavior.

\textbf{Clustering Analysis and Downstream Implementation:} As depicted in Figure \ref{fig:voting_clustersa}, the CFL algorithm identifies distinct clusters among the Voting dataset participants based on the same covariates in the analysis of \cite{gerber2009does}, including gender (\texttt{Bfemale}), age (\texttt{reportedage}), voting history (\texttt{Bvoted2001/02/04}), the source of sampling (\texttt{Bconsumer}), baseline political preference (\texttt{Bpreferrepub/dem}), and baseline exposure to similar types of papers (\texttt{Bgetsmag}). By excluding the treatment dummy, Figure \ref{fig:voting_clustersa} presents a coarsened covariate space constructed by above variables. 


Table \ref{tab:cluster_effects} summarizes the treatment effects of receiving Washington Post and Times, using the same regression specification in the original paper except that all covariates are replaced by cluster indicators. While the original paper concluded that, though not significant, receiving either type of paper lead to more support liberal attitude (both effects are negative), results in Table \ref{tab:cluster_effects} show that receiving the more liberal paper lead to more liberal post-treatment political attitude, and receiving the more conservative paper lead to more conservative political attitude. Despite this difference, which signifies different parts of the variations in data that the constructed macrostates and the original microstates respectively account for, they are still very similar in relative magnitude and significance level of coefficients, especially in the case of receiving the more liberal paper.

\subsubsection{CFL as Dimensionality Reduction Technique:} 

\begin{lemma}
    Let $\mathcal{X}$ be the random vector for all relevant covariates such that the unconfoundedness assumption $D \perp\!\!\!\perp (Y(1),Y(0)) \mid X$ holds, then the observational coarsening, $M$, of the covariate space of $\mathcal{X}$ by CFL also satisfy $D \perp\!\!\!\perp Y(\cdot) \mid M$
\end{lemma}

\begin{proof}
With the unconfoundedness assumption, the independence of the potential outcomes, $(Y(1),Y(0))$, and the treatment $D$ conditional on all confounders $\mathcal{X}$, which is a vector of all relevant covariates, ensures the recovery of the average treatment effect. We can formally write the unconfoundedness assumption as:
\begin{align*}
    D \perp\!\!\!\perp (Y(1),Y(0)) \mid X
\end{align*}
Recall that CFL conducts observational partition of the covariate space by the following equivalence relation: for any $x_i, x_j \in \mathcal{X}$:
\[
x_1 \sim x_2 \iff \forall y \in \mathcal{Y}, \ P\left(Y \mid X = x_i\right) = P\left(Y \mid X = x_j\right).
\]
Suppose $\{x_i\}_{i=1}^{\infty}$ is the sequence of all possible values in $\mathcal{X}$, and $\{x_{ik}\}_{k=1}^{n}$ is a subsequence that includes all the representatives of $n$ equivalent classes. Define a new random vector $M$ such that $M = \begin{bmatrix}
\1[x \in [x_{i1}]] & \1[x \in [x_{i2}]] & \cdots & \1[x \in [x_{in-1}]]
\end{bmatrix}^\top.$
This random vector is a coarsening of the original covariate space, and all the equivalence classes are mutually exclusive and jointly exhaustively. In other words, if $x \in [x_{i1}]$, then $\1[x \in [x_{i1}]] = 1$ and $\1[x \in [x_{ik}]] = 0$ for $k \neq 1$. Within each equivalence class $ [x_{ik}] $, the conditional distribution of $ Y(\cdot) $ given $ X $ is constant. This implies $Y(\cdot)$ is independent of $\mathcal{X}$ given $M$:
\begin{align*}
    P(Y(\cdot) \mid X, M) = P(Y(\cdot) \mid M) \Rightarrow Y(\cdot) \perp\!\!\!\perp X \mid M.
\end{align*}

By the law of total probability for conditional probability in continuous case: 
\begin{align*}
    P(Y(\cdot) \mid D, M) = \int P(Y(\cdot) \mid D, X, M) \, P(X \mid D, M) \, dX
\end{align*}
Since $Y(\cdot) \perp\!\!\!\perp D \mid X $, we have $P(Y(\cdot) \mid D, X, M) = P(Y(\cdot) \mid X, M) $. Since $Y(\cdot) \perp\!\!\!\perp X \mid M $, we have $ P(Y(\cdot) \mid X, M) = P(Y(\cdot) \mid M) $. Additionally, since $ P[Y(\cdot) \mid M] $ does not depend on $ X $, it can be factored out of the integral. Substituting all these into the integral:
\begin{align*}
    P(Y(\cdot) \mid D, M) = P(Y(\cdot) \mid M) \int P(X \mid D, M) \, dX.
\end{align*}
    The integral $ \int P(X \mid D, M) \, dX $ equals 1 by the definition of probability distributions. Therefore, this shows that $ Y(\cdot) $ is independent of $ D $ given $ M $:
\begin{align*}
    P(Y(\cdot) \mid D, M) = P(Y(\cdot) \mid M) \Rightarrow  D \perp\!\!\!\perp Y(\cdot) \mid M.
\end{align*}
\end{proof}

Therefore, the coarsening of the covariate space does not affect the conditional independence between the outcome and treatment, which then ensures that selection bias disappears:
\begin{align*}
    &\mathbb{E}[\mathbb{E}[Y \mid D = 1, M] - \mathbb{E}[Y \mid D = 0, M]] \\
    =& \mathbb{E}[\underbrace{\mathbb{E}[Y(1)-Y(0) \mid D=1, M]}_{\text{ATT}}]
    + \mathbb{E}[\underbrace{\mathbb{E}[Y(0) \mid D = 1, M] - \mathbb{E}[Y(0) \mid D = 0, M]}_{\mathbb{E}[Y(0) \mid M]-\mathbb{E}[Y(0) \mid M] = 0}] \\
    =& \underbrace{\mathbb{E}[Y(1)-Y(0)\mid M]}_{\text{ATE}}
\end{align*} 
In other words, the treatment is still as if randomized after controlling for all the macrostates created by the CFL algorithm.

\begin{remark}
    Under the unconfoundedness assumption, the observational partition reduces the dimensionality of the covariate space while preserving the causal structure in the data.
\end{remark}
 Let there are $p$ continuous covariates in $\mathcal{X}$, which would have a dimension of $\mathbb{R}^p$. By assumption of CFL, we have finite macrostate, or equivalence classes in the original covariate space, so the coarsened covariate space, $\mathcal{M}$, will have a dimension of $2^{n-1}$ where $n$ is finite as each random variable in the random vector $M$ is binary.

\section{Conclusion}
\label{sec:conclusion}

Our work extends the theoretical foundation of CFL for macrostate construction. To handle continuous variables, we developed a principled binning strategy and introduced extended definitions of the observational and causal partitions. Our central theoretical contribution, the Extended Causal Coarsening Theorem, demonstrates that under mild conditions the causal partition is almost surely a coarsening of the observational partition—even when outcomes are discretized into a finite number of bins.  Our empirical analyses of social science datasets show that CFL-derived macrostates effectively reduce dimensionality, uncover heterogeneous treatment effects, and preserve essential causal structures.  

\section*{Acknowledgments}

We gratefully acknowledge the insightful feedback provided by our colleagues, particularly Frederick Eberhardt and Milan Moss\'e, whose suggestions significantly enhanced this work.  

\bibliography{references}
\bibliographystyle{plainnat}

\clearpage

\appendix

\section{Proof of Theorem \ref{thm:cct-extended}}
\label{app:proof-sec-CFL-continuous}
To establish the Extended Causal Coarsening Theorem (CCT), we demonstrate that the set of joint distributions $P(X,Y,Z)$ violating the coarsening conditions has Lebesgue measure zero. This ensures that, almost surely, the causal partition $\Pi_c(\mathcal{X})$ is a coarsening of the observational partition $\Pi_o(\mathcal{X})$.
The proof is similar to the proof of Theorem \ref{thm:cct}.

We begin by defining the following probabilities:
\[
\begin{aligned}
\alpha_{k, x, z} &\triangleq P(a_k < y \le a_{k+1} \mid X = x, Z = z), \\
\beta_{x, z} &\triangleq P(X = x \mid Z = z), \\
\gamma_z &\triangleq P(Z = z).
\end{aligned}
\]
These definitions allow us to express the joint distribution as:
\[
P(X = x, a_k < y \le a_{k+1}, Z = z) = \gamma_z \cdot \beta_{x, z} \cdot \alpha_{k, x, z}.
\]

By fixing $\beta_{x, z}$, we determine the distribution of $\mathcal{X}$ given $\mathcal{Z}$, thereby fixing the confounding partition $\Pi_{P(X \mid Z)}(\mathcal{X})$. Similarly, fixing $\alpha_{k, x, z}$ specifies the causal relationship between $\mathcal{X}$ and $\mathcal{Y}$ given $\mathcal{Z}$, effectively fixing the causal partition $\Pi_c(\mathcal{X})$.

Next, consider the observational partition $\Pi_o(\mathcal{X})$, which groups together treatment variables $x_1$ and $x_2$ if and only if they induce the same distribution over $\mathcal{Y}$. Formally, $x_1 \sim x_2$ under $\Pi_o(\mathcal{X})$ if $\forall \text{ bin } (a_k,a_{k+1}]$:
\[
P(a_k < y \le a_{k+1} \mid x_1) = P(a_k < y \le a_{k+1} \mid x_2).
\]
Expanding $P(Y \mid X)$ using the law of total probability:
\begin{align*}
P(a_k < y \le a_{k+1} \mid X = x) &= \sum_{z} P(a_k < y \le a_{k+1} \mid X = x, Z = z) \cdot P(Z = z \mid X = x) \\
&= \sum_{z} P(a_k < y \le a_{k+1} \mid X = x, Z = z) \cdot \frac{P(X = x \mid Z = z) P(Z=z)}{P(X=x)} \\
&= \frac{\sum_{z} \gamma_z \cdot \beta_{x, z} \cdot \alpha_{k, x, z}}{\sum_{z} \gamma_{z} \cdot \beta_{x,z}}.
\end{align*}
Therefore, the equivalence condition $x_1 \sim x_2$ imposes:
\[
\frac{\sum_{z} \gamma_z \cdot \beta_{x_1, z} \cdot \alpha_{k, x_1, z}}{\sum_{z} \gamma_{z} \cdot \beta_{x_1,z}} = \frac{\sum_{z} \gamma_z \cdot \beta_{x_2, z} \cdot \alpha_{k, x_2, z}}{\sum_{z} \gamma_{z} \cdot \beta_{x_2,z}}.
\]
Rearranging terms leads to the polynomial constraints:
\[
\sum_{z_1,z_2} \gamma_{z_1}\gamma_{z_2} \left( \beta_{x_1, z_1} \cdot \alpha_{k, x_1, z_1} \cdot \beta_{x_2, z_2} - \beta_{x_2, z_1} \cdot \alpha_{k, x_2, z_1} \cdot \beta_{x_1,z_2} \right) = 0 \quad \forall \text{ bin } (a_k,a_{k+1}].
\]
These constraints define a system of polynomial equations in the parameters $\gamma_z$, $\beta_{x, z}$, and $\alpha_{k, x, z}$.

The set of joint distributions $P(X,Y,Z)$ that violate the condition $\Pi_c(\mathcal{X})$ being a coarsening of $\Pi_o(\mathcal{X})$ corresponds to the solutions of these polynomial equations. In the space of all possible joint distributions, these equations define an algebraic variety of lower dimension. 

We then verify that the polynomial constraints are non-trivial, i.e., not all $\{\gamma\}_{z \in \mathcal{Z}}$ satiafies these constraints. We first consider $\gamma_z = 1 / K$ for all $z$. If such $\{\gamma\}_{z \in \mathcal{Z}}$ do not satisfy the constraints, then the constraints are non-trivial. If the constraints are satisfied in this case, there exists $z_1^+,z_2^+$ such that $\left( \beta_{x_1, z_1^+} \cdot \alpha_{k, x_1, z_1^+} \cdot \beta_{x_2, z_2^+} - \beta_{x_2, z_1^+} \cdot \alpha_{k, x_2, z_1^+} \cdot \beta_{x_1,z_2^+} \right)$ is positive. And there also exists $z_1^-,z_2^-$ such that $\left( \beta_{x_1, z_1^-} \cdot \alpha_{k, x_1, z_1^-} \cdot \beta_{x_2, z_2^-} - \beta_{x_2, z_1^-} \cdot \alpha_{k, x_2, z_1^-} \cdot \beta_{x_1,z_2^-} \right)$ is negative. Without loss of generality, we assume that $z_1^+ \neq z_1^-$. Then we set $\gamma_z=1 / K$ for all $z \notin \{z_1^+, z_1^-\}$ and set $\gamma_{z_1^+}=\frac{3}{2K}$ and $\gamma_{z_1^-}=\frac{1}{2K}$. This is an example that violates the constraints.

Since the polynomial constraints are non-trivial, this variety has Lebesgue measure zero. Consequently, the subset of distributions $P(X,Y,Z)$ that do not satisfy the coarsening condition is of measure zero. Thus, with probability one (in the sense of Lebesgue measure), the causal partition $\Pi_c(\mathcal{X})$ is a coarsening of the observational partition $\Pi_o(\mathcal{X})$.

\section{Appendix: The Redlining Dataset}
\label{sec:redlining-data}

The Historic Redlining Indicator (HRI) is a measure of the mortgage investment risk of neighborhoods across the nation based on the residential security grades provided by The Home Owners' Loan Corporation (HOLC). A higher HRI score means greater redlining of the census tract. Using the 2010 and 2020 HRI datasets, we merged them with the US census dataset to investigate the impact of redlining on multiple socioeconomic outcomes.

\textbf{Clustering on Balanced Pseudo-Population after PSM:} Though randomized treatment is necessary for analysis using CFL, we can balance the data and create a pseudo-population using the propensity score matching. We define the treated group as those census tracts with an increase in the intensity of redlining, the HRI score, from 2010 to 2020, and untreated if the HRI score remains the same or decreases. Out of 11348 census tracts (an intersection of 2010 and 2020 census tracts), 487 are treated.  Currently, we choose two outcomes, changes in the proportion below the poverty line and changes in the median house value (in 2020 dollars), and 9 baseline covariates in 2010, including pubic school enrollment rate, average income, unemployment rate, the proportion of male, proportion of 3 racial groups (black, Asian, white), total housing units, and proportion of bachelor degrees, for the analysis. We used PSM with the nearest neighbor matching without replacement to create a pseudo-population such that each treated unit has a matched untreated unit (964 units). Figure \ref{fig:psm_balance} and \ref{fig:psm_propen_dis} present the balance of covariates and distribution of the propensity to treatment after propensity score matching. The estimated ATE and bootstrapped SE for two outcomes are 0.0028 [0.0128] (poverty rate) and -22599 [14928] (median house value). In other words, assume selection on observables, increase in redlining of a census tract increase the proportion of families below the poverty line, and lower the house value in the area, though both are insignificant.

After implementing CFL on the pseudo-population, there are two major observations. From Figure \ref{fig:redlining_3d}, we can see that, since the treatment is insignificant in both cases, CFL does not separate clusters primarily based on the treatment. However, CFL preserves the balance of treatment within each cluster. Though the distribution of the treated units and untreated units within each cluster is still similar to each other, the level of certain covariates is different between clusters, e.g. some are more skewed/sparse than others, indicating that different compositions of covariates across clusters are still related to the level of the outcome as shown in figure \ref{fig:cluster-comparison}.

\section{Appendix: Implementation of Binning Technique}
\label{sec:binning-implementation}


Here, we implement quantile binning, which is also known as equal frequency binning, in the example of the NSW dataset. We first apply an arbitrary number of bins, for example, 10 bins, by assigning a value to all observations in the same bin. We repeat the same CFL algorithm with the only change being the discretization of the outcome variable. Note that the variable is only discretized for the training of the CFL algorithm and construction of the macrostates; the original values of the variable should be retained for the rest of the analyzes. 

When the number of clusters is two, we almost replicate the distribution as in Figure \ref{fig:heterogeneity_clusters}, indicating that the amount of information on the causal relationship between the treatment, covariate, and outcome is not significantly compromised when we have a small number of clusters. When we increase the number of clusters to four, the results are less replicable as the CFL algorithm does not categorize a group of all untreated observations as in Figure \ref{fig:heterogeneity_clusters}. While we increase the number of clusters, during which the CFL algorithm creates more subgroups that capture more nuanced interaction between variables, a smaller number of bins eliminates some essential information entailed in a larger number of clusters. 

To verify this, we increase the number of bins and observe which bins we can obtain a distribution close enough to the one without the binning. If such a number of bins exist, it suggests that the same causal information relationship could be preserved while we discretize the variable and do not violate the assumption of discrete macrostate. As Table \ref{tab:min_perc_treated} shows, when we increase the number of bins, there are some cases where the CFL algorithm captures the group with only untreated units (value 0 in the table). As the algorithm includes randomness when constructing macrostates, the likelihood of successful construction of this particular macrostate increases as the number of bins increases.

\section{Appendix: Tables and Figures}

\setcounter{table}{0}
\setcounter{figure}{0}
\renewcommand{\thetable}{A\arabic{table}}
\renewcommand{\thefigure}{A\arabic{figure}}

\begin{table}[ht]
\centering
\caption{NSW Regression Results: Heterogeneity Identification}
\resizebox{0.6\linewidth}{!}{
\begin{tabular}{lcccc}
\toprule
\toprule
\textbf{Variable} & \textbf{Coef.} & \textbf{Std. Err.} & \textbf{t} & \textbf{P} $>$ $|\textbf{t}|$ \\
\midrule
Intercept        & 2790.36  & 476.31  & 5.85  & 0.00  \\
treat            & -409.66  & 742.27  & -0.55 & 0.58  \\
$\text{age}_{\text{dummy}}$      & -1670.13 & 721.93  & -2.31 & 0.02  \\
$\text{age}_{\text{dummy}}\times \text{treat}$ & 2889.34  & 1125.78 & 2.56  & 0.01  \\
\bottomrule
\bottomrule
\end{tabular}
}
\label{tab:regression_results}
\end{table}

\begin{table}[h]
    \centering
    \caption{Effects of Washington Post and Times for Different Clusters}
    \label{tab:cluster_effects}
    \resizebox{0.52\linewidth}{!}{%
    \begin{tabular}{c@{\hskip 0pt}*{2}{>{\centering\arraybackslash}p{4cm}}}
        \toprule
        \toprule
        \# Cluster & Post Effect & Times Effect \\
        \midrule
        3  & -0.0733 & 0.0112 \\
           & (0.0475) & (0.0469) \\
           & [0.1230] & [0.8118] \\ [2ex]
        4  & -0.0742 & 0.0007 \\
           & (0.0472) & (0.0465) \\
           & [0.1164] & [0.9880] \\ [2ex]
        6  & -0.0684 & -0.0022 \\
           & (0.0470) & (0.0465) \\
           & [0.1458] & [0.9618] \\ [2ex]
        8  & -0.0684 & 0.0090 \\
           & (0.0473) & (0.0467) \\
           & [0.1488] & [0.8478] \\ [2ex]
        12 & -0.0646 & 0.0018 \\
           & (0.0474) & (0.0468) \\
           & [0.1727] & [0.9688] \\
        \bottomrule
        \bottomrule
    \end{tabular}
    }
    \parbox{0.52\textwidth}{%
        \footnotesize 
         P-values are presented in square brackets and standard errors are presented in parentheses.}
\end{table}

\begin{table}[h]
    \centering
    \caption{Minimum Percentage of Treated Units in a Cluster by Bins}
    \resizebox{0.38\linewidth}{!}{%
    \begin{tabular}{c@{\hskip 0pt}*{2}{>{\centering\arraybackslash}p{5cm}}}
        \toprule
        \toprule
        \textbf{\# Bins} & \textbf{Min \% of Treated} \\ 
        \midrule
        10  & 0.321  \\
        27  & 0.129  \\
        45  & 0.336  \\
        62  & 0.230  \\
        88  & 0.292  \\
        131 & 0.198  \\
        174 & 0.320  \\
        218 & 0.000  \\
        262 & 0.311  \\
        305 & 0.102  \\
        348 & 0.207  \\
        392 & 0.327  \\
        435 & 0.000  \\
        \bottomrule
        \bottomrule
    \end{tabular}}
    \label{tab:min_perc_treated}
\end{table}

\begin{figure}[h]
\caption{Clustering of Voting Dataset Participants Based on Demographics, Baseline Political Preference, and Historical Turnout Record}
    \centering
    \includegraphics[width=0.8\linewidth]{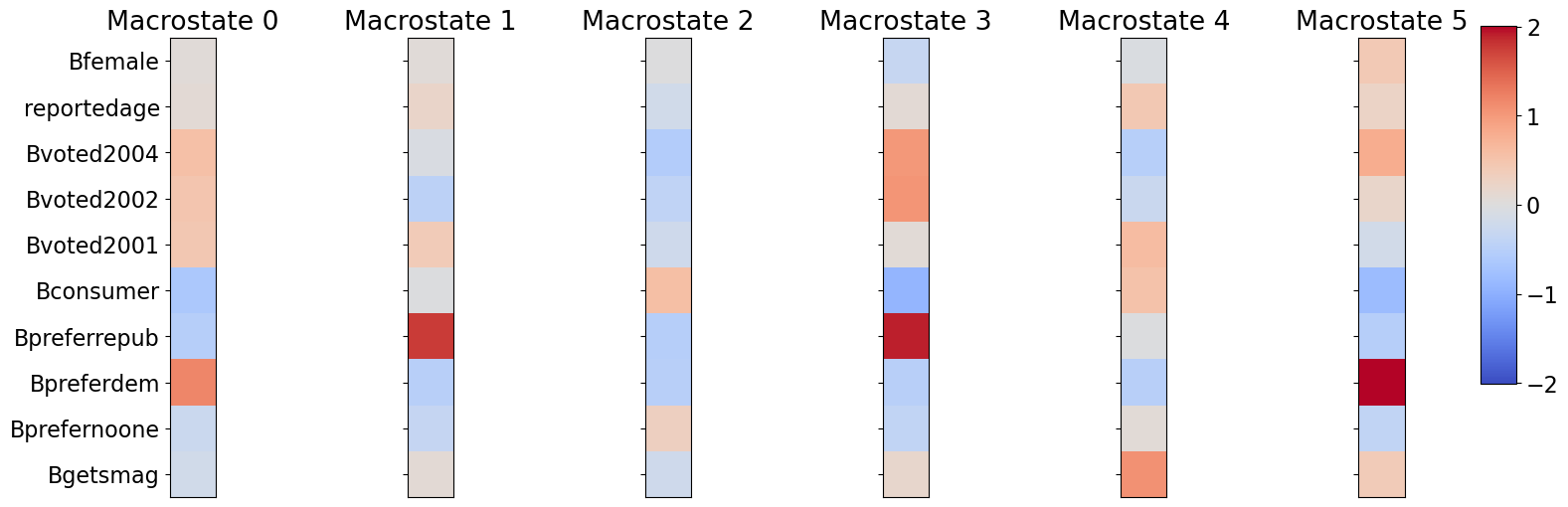}
    \label{fig:voting_clustersa}
\end{figure}

\begin{figure}[h]
\caption{Clustering of Redlining Dataset}
    \centering
    \subfloat[House value clustering]{%
        \includegraphics[scale=0.33]{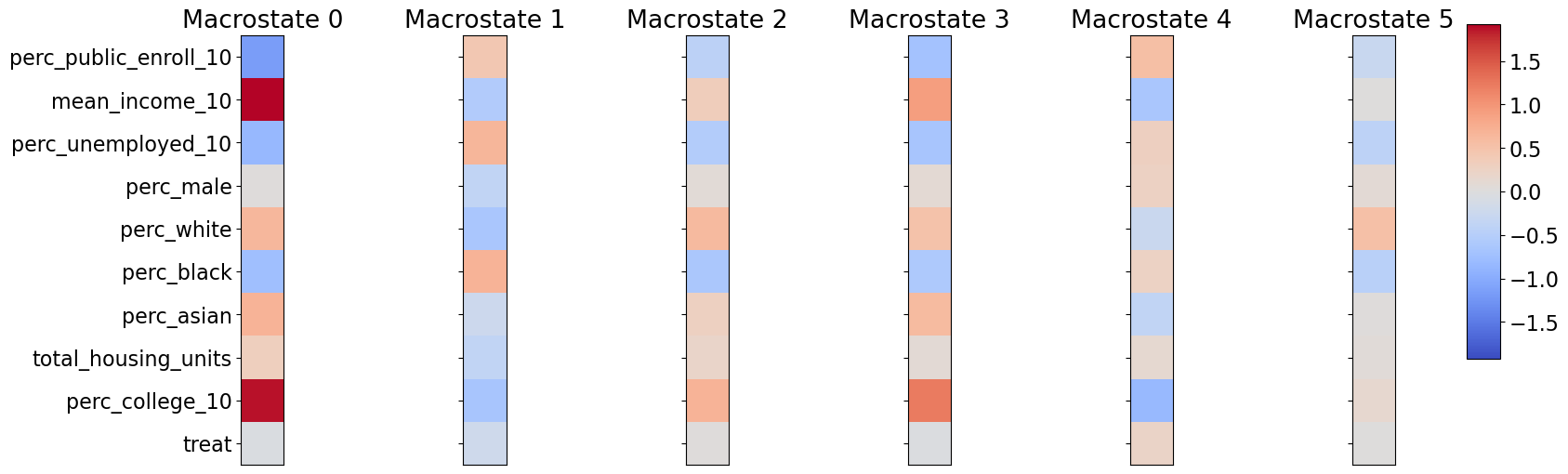}
        \label{fig:cluster-poverty}
    }
    \hspace{1cm} 
    \subfloat[Poverty rate clustering]{%
        \includegraphics[scale=0.33]{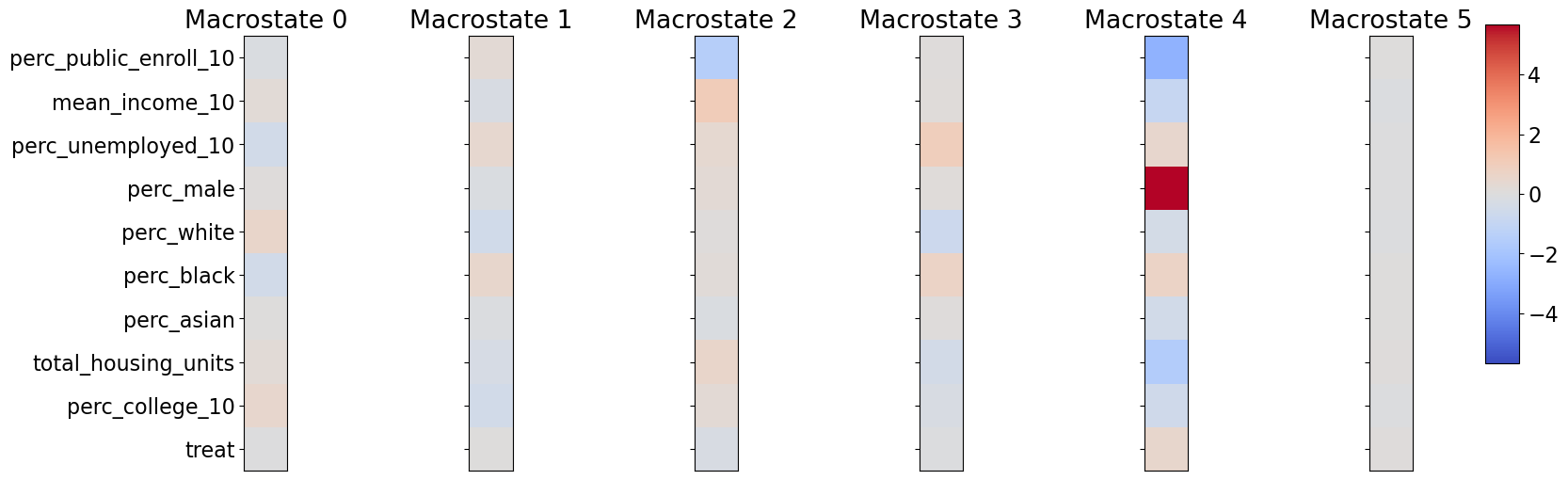}
        \label{fig:cluster-value}
    }
    \label{fig:cluster-comparison}
\end{figure}

\begin{figure}[h]
    \centering
    \begin{minipage}{0.58\textwidth}
        \centering
        \includegraphics[width=\linewidth]{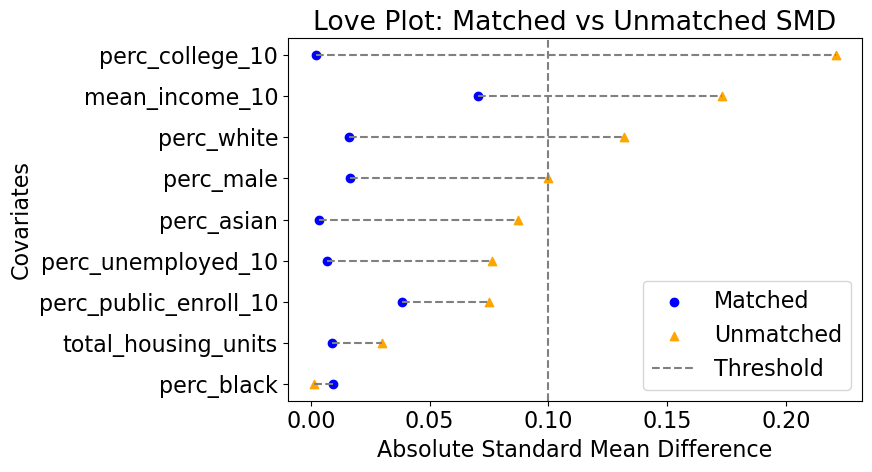}
        \caption{Balance Check}
        \label{fig:psm_balance}
    \end{minipage}
    \hfill
    \begin{minipage}{0.40\textwidth}
        \centering
        \includegraphics[width=\linewidth]{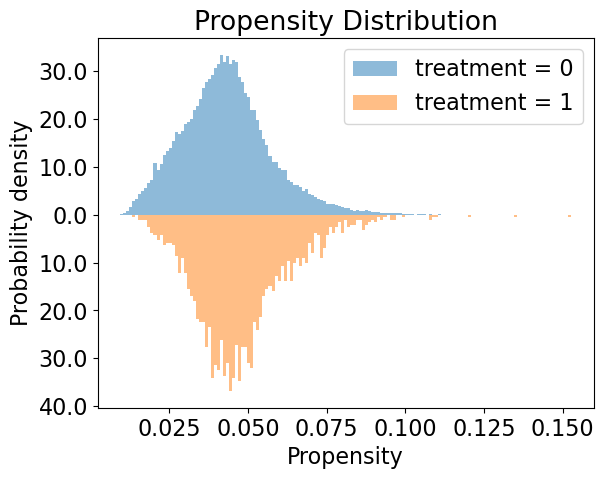}
        \caption{Propensity to Treatment}
        \label{fig:psm_propen_dis}
    \end{minipage}
\end{figure}

\begin{figure}[ht]
    \centering
    \includegraphics[scale=.26]{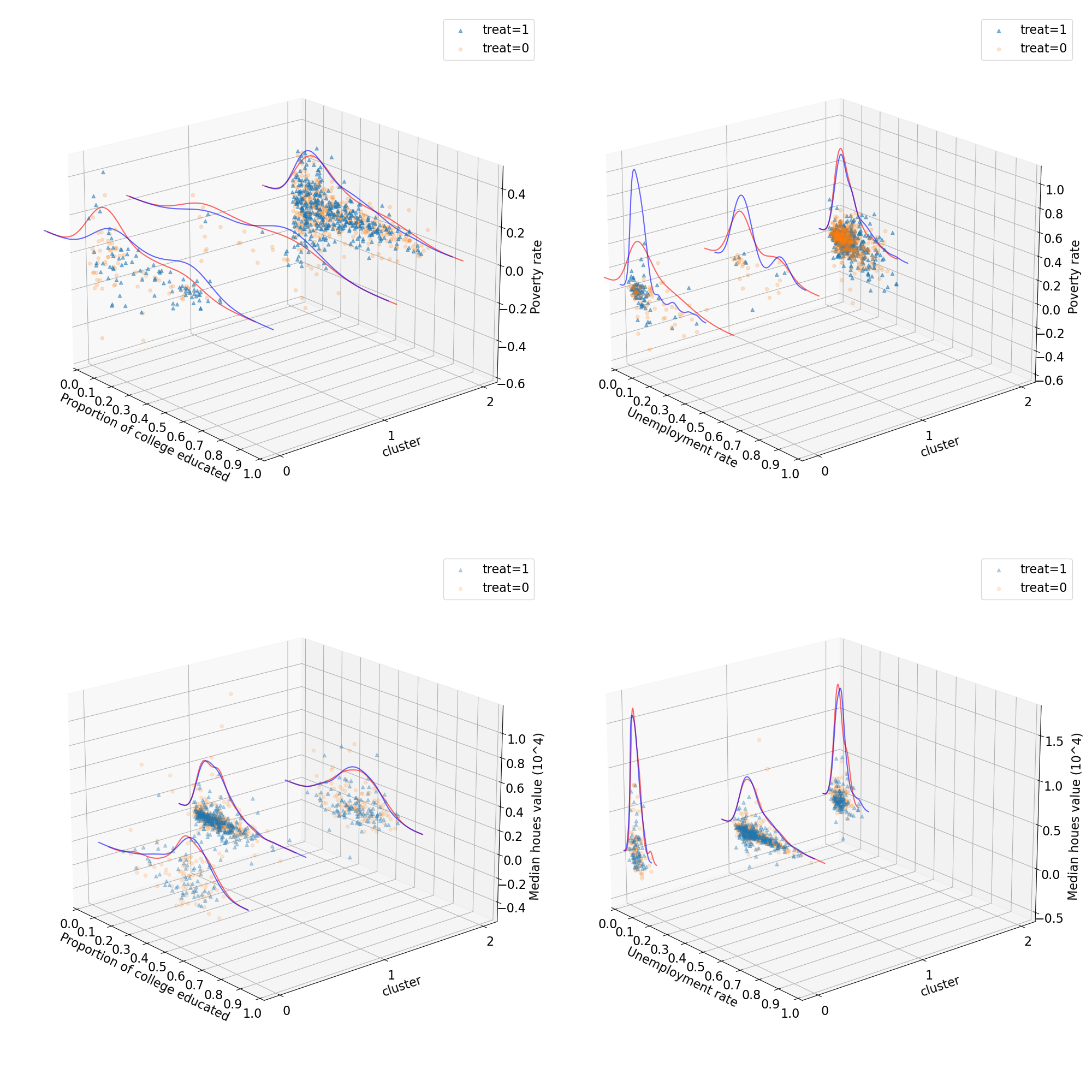}
    \caption{Distribution of Treated and Untreated Units
Across Clusters}
    \label{fig:redlining_3d}
\end{figure}

\end{document}